\documentclass[submission,copyright,creativecommons]{eptcs}
\providecommand{\event}{TARK 2025} 

\usepackage{iftex}

\ifpdf
  \usepackage{underscore}         
  \usepackage[T1]{fontenc}        
\else
  \usepackage{breakurl}           
\fi

\title{Muddy Waters}
\author{
Hans van Ditmarsch
\institute{University of Toulouse, CNRS, IRIT}
\email{hansvanditmarsch@gmail.com}
}
\def\titlerunning{Muddy Waters}
\def\authorrunning{Hans van Ditmarsch}

\usepackage{amsfonts, amsmath, amssymb}
\usepackage{color}
\usepackage{tikz,url}
\usetikzlibrary{calc,shapes}

\newcommand{\eq}{\leftrightarrow}

\newcommand{\imp}{\rightarrow}
\newcommand{\Imp}{\Rightarrow}

\newcommand{\et}{\wedge}
\newcommand{\vel}{\vee}
\newcommand{\Et}{\bigwedge}
\newcommand{\Vel}{\bigvee}

\renewcommand{\phi}{\varphi}
\newcommand{\union}{\cup}
\newcommand{\Union}{\bigcup}
\newcommand{\inter}{\cap}
\newcommand{\Inter}{\bigcap}

\newcommand{\dia}[1]{\langle{#1}\rangle}

\newcommand{\M}{\widehat{K}}

\newcommand{\Nat}{\mathbb N}
\newcommand{\Naturals}{\Nat}

\usepackage{newproof}

\newtheorem{theorem}{Theorem}

\newtheorem{proposition}[theorem]{Proposition}

\newcommand{\weg}[1]{}

\newcommand{\lbr}{[\![}
\newcommand{\rbr}{]\!]}
\newcommand{\I}[1]{\lbr #1 \rbr}

\newcommand{\solvable}{\ensuremath{\mathbf{solvable}}}
\newcommand{\bell}{\ensuremath{\mathbf{bell}}}

\newcommand{\last}{\ensuremath{\mathbf{last}}}
\newcommand{\muetzen}{\ensuremath{\mathbf{muetzen}}}

\usepackage{bbold}

\newcommand{\notlast}{\ensuremath{\mathbf{nlast}}}

\usepackage{bbold}
\usepackage{bm}

\begin{document}
\maketitle

\begin{abstract}
In the 2013 Advent calender of the Berlin Mathematics Research Center MATH+, Gerhard Woeginger presents a novel hat problem with an uncommon initial announcement. Although the information given is insufficient for the hat bearers to learn their colour, they are informed that the colours have been chosen so that they can learn their colour. We formalize this announcement in public announcement logic and in an extension of public announcement logic with fixpoints.
\end{abstract}

\section{A new hats riddle}

\subsection{History of the muddy children puzzle}

The Muddy Children puzzle and similar hats riddles have delighted puzzle book enthusiasts since at least the 1930s \cite{kraitchik:1942,littlewood:1953,vantilburg:1956,gamowetal:1958,smullyan:1982} (with older roots going back to at least the early 19th century \cite{rabelais:1823}, see \cite{hvdetal.puzzle:2015,hvdetal.jlc:2007} for further notes on its history) and has also fuelled the development of artificial intelligence and of dynamic epistemic logic \cite{mccarthy:1978,barwise:saos,mosesetal:1986,Emde:84,plaza:1989,hvdetal.del:2007}.

From an epistemic logical perspective such knowledge puzzles tend to follow a certain pattern. First, the problem solver has to represent the initial uncertainty of the agents in a multi-agent Kripke model. This initial model represents the background knowledge. Second, for the agents this may be insufficient information in order to solve their uncertainty by ignorance or knowledge statements (for example, when their ignorance is already common knowledge). Some extra information (an initial announcement) is then needed in order to make the agents' ignorance statements informative. Third, the agents then keep going announcing their ignorance until the problem is solved. In the case of muddy children: (i) the initial model encodes that the children can only see the mud on each other's foreheads (it is common knowledge that no child knows whether it is muddy); (ii) father then announces that  at least one child is muddy; (iii) the children keep `announcing' their ignorance of muddy in rounds until they know whether they are muddy. To synchronize the children's behaviour, a round is characterized by father repeating his request, or clapping in his hands, or ringing a bell. Stage (ii) to `get the induction going' is not always there, for example not in Consecutive Numbers \cite{littlewood:1953} and not in Sum and Product \cite{freudenthal:1969}.

\subsection{M\"utzen}

The Berlin Mathematics Research Center MATH+ annually presents an Advent calender with mathematical riddles, of which the solutions are subsequently given after Christmas. In the 2013 calender problem number 10, by Gerhard Woeginger, is called M\"utzen (`Hats') \cite{woegingermuetzen:2013}. It is a variation of the muddy children and hats puzzles with many hats and many announcements, in a Christmas setting where Santa informs his Little Helpers (`Wichtel', Gnomes) of the solution requirement. It is accompanied by a beautiful illustration implicitly challenging the reader to determine the famous mathematicians depicted there as gnomes with coloured hats (Ren\'e Descartes, Leonhard Euler and Kurt G\"odel are among them). M\"utzen is given in the Appendix in its original playful and beautiful German version and in English translation.

In the riddle's description we can easily see the three stages needed to solve it. First, the initial model encodes 126 gnomes only seeing the colours of others' hats. Nothing is known about these colours. Second, Santa makes a, in this case, curious initial announcement. We call this announcement \solvable.
\begin{quote} {\em I chose the hat colours very carefully so that each of you can actually determine their colour through thinking during the game.} \hfill $(\solvable)$
\end{quote}
Third, the bell starts ringing, where at each ring those leave who know their colour, until finally all are having tea and cake. 

What does Santa's initial announcement \solvable\ say? In muddy children and hat problems it is always given what the set of colours is, or that children can be muddy or not, which can be seen as the set of colours consisting of black and white. In M\"utzen we only know that the hats are coloured. Just as in the muddy children problem and in other hat problems, without additional information that the problem solver provides (`at least one child is muddy') they never learn anything from the bell being rung and such. 

The information content of \solvable\ is much harder to grasp than that of father's announcement that at least one child is muddy. The latter can be easily verified, unlike the former. Atto gives an important hint when he says to Santa `For example, if each of us had a different colour of hat, then no one would be able to figure out what colour it is'. Let us be precise about Atto's observed inability to solve the riddle. 

If all gnomes wear a different unique colour, any gnome $i$ wearing colour $c$ also considers it possible that its hat has colour $c'$ for some colour $c' \neq c$ that is also not worn by any of the other gnomes. Therefore no gnome knows the colour of its hat and therefore no one will leave the room when the bell rings. And also not in the next iteration. Not ever. The later distribution of hats reveals that Atto cannot have seen $125$ all different colours. But this does not matter: as long as any gnome considers it possible that its own colour is unique, it considers it possible that this is any of two unseen colours and will therefore not leave the room when the bell rings. Therefore, it cannot have a unique colour, and therefore no gnome can have a unique colour, and therefore of each colour occurring there must be at least two hats. So that all the colours it sees, are all the colours there are. This demonstrates that the announcement \solvable\ is therefore {\em at least as} informative as that of the announcement: 
\begin{quote} {\em For each coloured hat worn by a gnome there must be another gnome wearing a hat with the same colour. \hfill $(\solvable')$} \end{quote} We call that announcement $\solvable'$. This is indeed the spark that gets the induction going, so that finally all can determine their colour. That demonstrates that \solvable\ is {\em at most as} informative as $\solvable'$. Therefore, \solvable\ and $\solvable'$ have the same information content. 

Let us sketch how the iterations proceed that solve the problem:

At the first iteration all gnomes seeing only one gnome with a hat of a particular colour conclude that they must also have that colour, and leave the room. At the second iteration all (only) seeing two gnomes with a hat of a particular colour, leave. And so on. If in iteration $n$ no one sees $n$ gnomes of some colour, no one leaves the room; which is on condition, fulfilled in the riddle, that different colours can be seen. For example, let there be only $12$ gnomes, with $2$ white and $2$ black and $8$ rose hats. Then at ring $1 $ the gnomes with white and black hats leave, but already at ring $2$ the gnomes with rose hats leave. \emph{It does not take more rounds.} But if there had been $10$ more pink and $10$ more yellow hatted gnomes, there would have been no such shortening of rounds for the $8$ rose hatted gnomes. Because then: at ring $1$ white and black leave, at rings $2$ to $6$ no one leaves, at ring $7$ the gnomes with rose hats leave, at ring $8$ noone leaves, and at ring $9$ the gnomes with pink and yellow hats leave. In the M\"utzen riddle shortening of rounds does not play a role, as there are two groups of the same maximum size that step forward in the last round. We do not know if this was by design, in order to avoid this extra modelling complication.

\subsection{Solvable and unsolvable variations of muddy children} \label{solvableunsolvable}

Santa's announcement \solvable\ reduces the set of all colour distributions to the largest subset $X$ such that all gnomes will leave the room after some finite iteration of ringing the bell, which informally means that this is a \emph{fixpoint} with respect to solving the problem: any set of colour distributions of which $X$ is a strict subset will not solve the problem, although there are sets of colour distributions that are strict subsets of $X$ and that solve the problem. Furthermore, no actual colours are mentioned in \solvable. 

Before we delve into the logic and formally introduce fixpoints, let us investigate fixpoints in the simpler setting of the muddy children problem, with colours black (muddy) and white (clean), and for three children only. First, the standard solution: given that children only see the foreheads of others, father says that at least one child is muddy, after which at each round the children who know whether they are muddy step forward. This takes at most three rounds. 

In Figure~\ref{fig.muddy}.i we visualize the execution of this branching protocol in a linear way, as iterated refinement of epistemic models, because all branching is a consequence of mutually exclusive public announcements. We also see father's initial announcement not as a model restriction but as a model refinement. The cube and its subsequent submodels represent the uncertainty of children $a,b,c$. The states (or worlds) are named with triples of digits indicating whether child $a,b,c$ in this order are muddy (1) or clean (0). Agent labelled links between states indicate that the states are indistinguishable for those agents. The double arrows indicate model updates. Each arrow linking two models represents non-deterministic choice between mutually exclusive alternatives, where the first arrow is choice between two alternatives, often known as a \emph{test} (whether), denoted with `?'. 
The final `arrow' is a refinement that does not change the model and is therefore denoted $=$ instead of $\Imp$.

\paragraph{Different ways to solve muddy children} 
We will not explain this standard solution of the muddy children problem for the umptieth time. What we wish to point out is that apart from this standard initial refinement into a seven-state and a one-state model, in some other initial refinements  all children also eventually learn whether they are muddy, whereas in yet other initial refinements this is not the case. We will not systematically review all these restrictions, but merely make some pregnant observations: 

First, \emph{any} initial restriction to seven states results in all children learning whether they are muddy, and by the exact same updates. 

Second, \emph{some} restrictions result in some children learning whether they are muddy, but not all, or not always. A typical one is the refinement where we separate an `edge', such as $000$---$a$---$100$ from its complement. In both parts $a$ will never learn whether it is muddy. There is too much symmetry in the restriction. 

Third, it is not so clear how to see any of the restrictions where all learn whether they are muddy as a fixpoint. In particular, the announcement $\solvable'$ that there are at least two hats of every colour is trivialized because there are fewer colours $(2)$ than agents $(3)$. It is the restriction to $\{000,111\}$. It is not a fixpoint, e.g.\  Figure~\ref{fig.muddy}.i and \ref{fig.muddy}.ii both also contain $111$. It is not even colour-blind: there are other restrictions that are also invariant for any permutation of colours over hats. For example, any other opposite states such as $\{011,100\}$. 

Fourth, did you, reader, ever realize that there are really two different ways to solve the muddy children problem? Take Figure~\ref{fig.muddy}.ii where in the three-state restriction $\{000,100,110\}$ after two rounds of the $\bell$ finally all children have stepped forward, whereas in the five-state complement restriction this takes three rounds instead. Now compare this to Figure~\ref{fig.muddy}.iii where the updates are a function of the announcement that nobody knows whether they are muddy (or its negation), later formalized as $\bell_\emptyset$. The updates are different! How come?


%
%
%
%
%
\begin{figure}
\[\begin{array}{lllllllll}
\raisebox{47pt}{$(i)$} &
\scalebox{.45}{
\framebox{
\begin{tikzpicture}[z=0.35cm]
\node (000) at (0,0,0) {$000$};
\node (001) at (0,0,3) {$001$};
\node (010) at (0,3,0) {$010$};
\node (011) at (0,3,3) {$011$};
\node (100) at (3,0,0) {$100$};
\node (101) at (3,0,3) {$101$};
\node (110) at (3,3,0) {$110$};
\node (111) at (3,3,3) {$111$};
\draw (000) -- node[fill=white,inner sep=1pt] {$a$} (100);
\draw (001) -- node[fill=white,inner sep=1pt] {$a$} (101);
\draw (011) -- node[fill=white,inner sep=1pt] {$a$} (111);
\draw (000) -- node[fill=white,inner sep=1pt] {$b$} (010);
\draw (001) -- node[fill=white,inner sep=1pt] {$b$} (011);
\draw (101) -- node[fill=white,inner sep=1pt] {$b$} (111);
\draw (000) -- node[fill=white,inner sep=1pt] {$c$} (001);
\draw (010) -- node[fill=white,inner sep=1pt] {$c$} (011);
\draw (100) -- node[fill=white,inner sep=1pt] {$c$} (101);
\draw (110) -- node[fill=white,inner sep=1pt] {$c$} (111);
\draw[draw=white,double=black,very thick] (010) -- node[fill=white,inner sep=1pt] {$a$} (110);
\draw[draw=white,double=black,very thick] (100) -- node[fill=white,inner sep=1pt] {$b$} (110);
\end{tikzpicture}
}
} 
& \raisebox{25pt}{$\stackrel {000?} \Imp$} &
\scalebox{.45}{
\framebox{
\begin{tikzpicture}[z=0.35cm]
\node (000) at (0,0,0) {$000$};
\node (001) at (0,0,3) {$001$};
\node (010) at (0,3,0) {$010$};
\node (011) at (0,3,3) {$011$};
\node (100) at (3,0,0) {$100$};
\node (101) at (3,0,3) {$101$};
\node (110) at (3,3,0) {$110$};
\node (111) at (3,3,3) {$111$};
%
\draw (001) -- node[fill=white,inner sep=1pt] {$a$} (101);
\draw (011) -- node[fill=white,inner sep=1pt] {$a$} (111);
%
\draw (001) -- node[fill=white,inner sep=1pt] {$b$} (011);
\draw (101) -- node[fill=white,inner sep=1pt] {$b$} (111);
%
\draw (010) -- node[fill=white,inner sep=1pt] {$c$} (011);
\draw (100) -- node[fill=white,inner sep=1pt] {$c$} (101);
\draw (110) -- node[fill=white,inner sep=1pt] {$c$} (111);
\draw[draw=white,double=black,very thick] (010) -- node[fill=white,inner sep=1pt] {$a$} (110);
\draw[draw=white,double=black,very thick] (100) -- node[fill=white,inner sep=1pt] {$b$} (110);
\end{tikzpicture} 
}
}
& \raisebox{25pt}{$\stackrel \bell \Imp$} &
\scalebox{.45}{
\framebox{
\begin{tikzpicture}[z=0.35cm]
\node (000) at (0,0,0) {$000$};
\node (001) at (0,0,3) {$001$};
\node (010) at (0,3,0) {$010$};
\node (011) at (0,3,3) {$011$};
\node (100) at (3,0,0) {$100$};
\node (101) at (3,0,3) {$101$};
\node (110) at (3,3,0) {$110$};
\node (111) at (3,3,3) {$111$};
%
\draw (011) -- node[fill=white,inner sep=1pt] {$a$} (111);
%
\draw (101) -- node[fill=white,inner sep=1pt] {$b$} (111);
%
\draw (110) -- node[fill=white,inner sep=1pt] {$c$} (111);
%
%
\end{tikzpicture} 
}
}
& \raisebox{25pt}{$\stackrel \bell \Imp$} &
\scalebox{.45}{
\framebox{
\begin{tikzpicture}[z=0.35cm]
\node (000) at (0,0,0) {$000$};
\node (001) at (0,0,3) {$001$};
\node (010) at (0,3,0) {$010$};
\node (011) at (0,3,3) {$011$};
\node (100) at (3,0,0) {$100$};
\node (101) at (3,0,3) {$101$};
\node (110) at (3,3,0) {$110$};
\node (111) at (3,3,3) {$111$};
%
%
%
%
%
\end{tikzpicture} 
}
}
& \raisebox{25pt}{$\stackrel \bell =$}
\end{array}\]
\[\begin{array}{lllllllll}
\raisebox{47pt}{$(ii)$} &
\scalebox{.45}{
\framebox{
\begin{tikzpicture}[z=0.35cm]
\node (000) at (0,0,0) {$000$};
\node (001) at (0,0,3) {$001$};
\node (010) at (0,3,0) {$010$};
\node (011) at (0,3,3) {$011$};
\node (100) at (3,0,0) {$100$};
\node (101) at (3,0,3) {$101$};
\node (110) at (3,3,0) {$110$};
\node (111) at (3,3,3) {$111$};
\draw (000) -- node[fill=white,inner sep=1pt] {$a$} (100);
\draw (001) -- node[fill=white,inner sep=1pt] {$a$} (101);
\draw (011) -- node[fill=white,inner sep=1pt] {$a$} (111);
\draw (000) -- node[fill=white,inner sep=1pt] {$b$} (010);
\draw (001) -- node[fill=white,inner sep=1pt] {$b$} (011);
\draw (101) -- node[fill=white,inner sep=1pt] {$b$} (111);
\draw (000) -- node[fill=white,inner sep=1pt] {$c$} (001);
\draw (010) -- node[fill=white,inner sep=1pt] {$c$} (011);
\draw (100) -- node[fill=white,inner sep=1pt] {$c$} (101);
\draw (110) -- node[fill=white,inner sep=1pt] {$c$} (111);
\draw[draw=white,double=black,very thick] (010) -- node[fill=white,inner sep=1pt] {$a$} (110);
\draw[draw=white,double=black,very thick] (100) -- node[fill=white,inner sep=1pt] {$b$} (110);
\end{tikzpicture}
}
} 
& \raisebox{25pt}{$\stackrel {3/5?} \Imp$} &
\scalebox{.45}{
\framebox{
\begin{tikzpicture}[z=0.35cm]
\node (000) at (0,0,0) {$000$};
\node (001) at (0,0,3) {$001$};
\node (010) at (0,3,0) {$010$};
\node (011) at (0,3,3) {$011$};
\node (100) at (3,0,0) {$100$};
\node (101) at (3,0,3) {$101$};
\node (110) at (3,3,0) {$110$};
\node (111) at (3,3,3) {$111$};
\draw (000) -- node[fill=white,inner sep=1pt] {$a$} (100);
\draw (001) -- node[fill=white,inner sep=1pt] {$a$} (101);
\draw (011) -- node[fill=white,inner sep=1pt] {$a$} (111);
%
\draw (001) -- node[fill=white,inner sep=1pt] {$b$} (011);
\draw (101) -- node[fill=white,inner sep=1pt] {$b$} (111);
%
\draw (010) -- node[fill=white,inner sep=1pt] {$c$} (011);
%
\draw[draw=white,double=black,very thick] (100) -- node[fill=white,inner sep=1pt] {$b$} (110);
\end{tikzpicture} 
}
}
& \raisebox{25pt}{$\stackrel \bell \Imp$} &
\scalebox{.45}{
\framebox{
\begin{tikzpicture}[z=0.35cm]
\node (000) at (0,0,0) {$000$};
\node (001) at (0,0,3) {$001$};
\node (010) at (0,3,0) {$010$};
\node (011) at (0,3,3) {$011$};
\node (100) at (3,0,0) {$100$};
\node (101) at (3,0,3) {$101$};
\node (110) at (3,3,0) {$110$};
\node (111) at (3,3,3) {$111$};
%
\draw (001) -- node[fill=white,inner sep=1pt] {$a$} (101);
%
\draw (101) -- node[fill=white,inner sep=1pt] {$b$} (111);
%
%
%
\end{tikzpicture} 
}
}
& \raisebox{25pt}{$\stackrel \bell \Imp$} &
\scalebox{.45}{
\framebox{
\begin{tikzpicture}[z=0.35cm]
\node (000) at (0,0,0) {$000$};
\node (001) at (0,0,3) {$001$};
\node (010) at (0,3,0) {$010$};
\node (011) at (0,3,3) {$011$};
\node (100) at (3,0,0) {$100$};
\node (101) at (3,0,3) {$101$};
\node (110) at (3,3,0) {$110$};
\node (111) at (3,3,3) {$111$};
%
%
%
%
%
\end{tikzpicture} 
}
}
& \raisebox{25pt}{$\stackrel \bell =$}
\end{array}\]
%
%
%
%
\[\begin{array}{lllllllll}
\raisebox{47pt}{$(iii)$} &
\scalebox{.45}{
\framebox{
\begin{tikzpicture}[z=0.35cm]
\node (000) at (0,0,0) {$000$};
\node (001) at (0,0,3) {$001$};
\node (010) at (0,3,0) {$010$};
\node (011) at (0,3,3) {$011$};
\node (100) at (3,0,0) {$100$};
\node (101) at (3,0,3) {$101$};
\node (110) at (3,3,0) {$110$};
\node (111) at (3,3,3) {$111$};
\draw (000) -- node[fill=white,inner sep=1pt] {$a$} (100);
\draw (001) -- node[fill=white,inner sep=1pt] {$a$} (101);
\draw (011) -- node[fill=white,inner sep=1pt] {$a$} (111);
\draw (000) -- node[fill=white,inner sep=1pt] {$b$} (010);
\draw (001) -- node[fill=white,inner sep=1pt] {$b$} (011);
\draw (101) -- node[fill=white,inner sep=1pt] {$b$} (111);
\draw (000) -- node[fill=white,inner sep=1pt] {$c$} (001);
\draw (010) -- node[fill=white,inner sep=1pt] {$c$} (011);
\draw (100) -- node[fill=white,inner sep=1pt] {$c$} (101);
\draw (110) -- node[fill=white,inner sep=1pt] {$c$} (111);
\draw[draw=white,double=black,very thick] (010) -- node[fill=white,inner sep=1pt] {$a$} (110);
\draw[draw=white,double=black,very thick] (100) -- node[fill=white,inner sep=1pt] {$b$} (110);
\end{tikzpicture}
}
} 
& \raisebox{25pt}{$\stackrel {3/5?} \Imp$} &
\scalebox{.45}{
\framebox{
\begin{tikzpicture}[z=0.35cm]
\node (000) at (0,0,0) {$000$};
\node (001) at (0,0,3) {$001$};
\node (010) at (0,3,0) {$010$};
\node (011) at (0,3,3) {$011$};
\node (100) at (3,0,0) {$100$};
\node (101) at (3,0,3) {$101$};
\node (110) at (3,3,0) {$110$};
\node (111) at (3,3,3) {$111$};
\draw (000) -- node[fill=white,inner sep=1pt] {$a$} (100);
\draw (001) -- node[fill=white,inner sep=1pt] {$a$} (101);
\draw (011) -- node[fill=white,inner sep=1pt] {$a$} (111);
%
\draw (001) -- node[fill=white,inner sep=1pt] {$b$} (011);
\draw (101) -- node[fill=white,inner sep=1pt] {$b$} (111);
%
\draw (010) -- node[fill=white,inner sep=1pt] {$c$} (011);
%
\draw[draw=white,double=black,very thick] (100) -- node[fill=white,inner sep=1pt] {$b$} (110);
\end{tikzpicture} 
}
}
& \raisebox{25pt}{$\stackrel {\bell_\emptyset?} \Imp$} &
\scalebox{.45}{
\framebox{
\begin{tikzpicture}[z=0.35cm]
\node (000) at (0,0,0) {$000$};
\node (001) at (0,0,3) {$001$};
\node (010) at (0,3,0) {$010$};
\node (011) at (0,3,3) {$011$};
\node (100) at (3,0,0) {$100$};
\node (101) at (3,0,3) {$101$};
\node (110) at (3,3,0) {$110$};
\node (111) at (3,3,3) {$111$};
\draw (000) -- node[fill=white,inner sep=1pt] {$a$} (100);
\draw (001) -- node[fill=white,inner sep=1pt] {$a$} (101);
%
\draw (101) -- node[fill=white,inner sep=1pt] {$b$} (111);
%
%
\draw[draw=white,double=black,very thick] (100) -- node[fill=white,inner sep=1pt] {$b$} (110);
\end{tikzpicture} 
}
}
& \raisebox{25pt}{$\stackrel {\bell_\emptyset?} =$} & \qquad \raisebox{25pt}{\LARGE$\times$} \qquad\qquad & \qquad\qquad\qquad
\end{array}\]
\caption{Different ways to solve three muddy children}
\label{fig.muddy}
\end{figure}
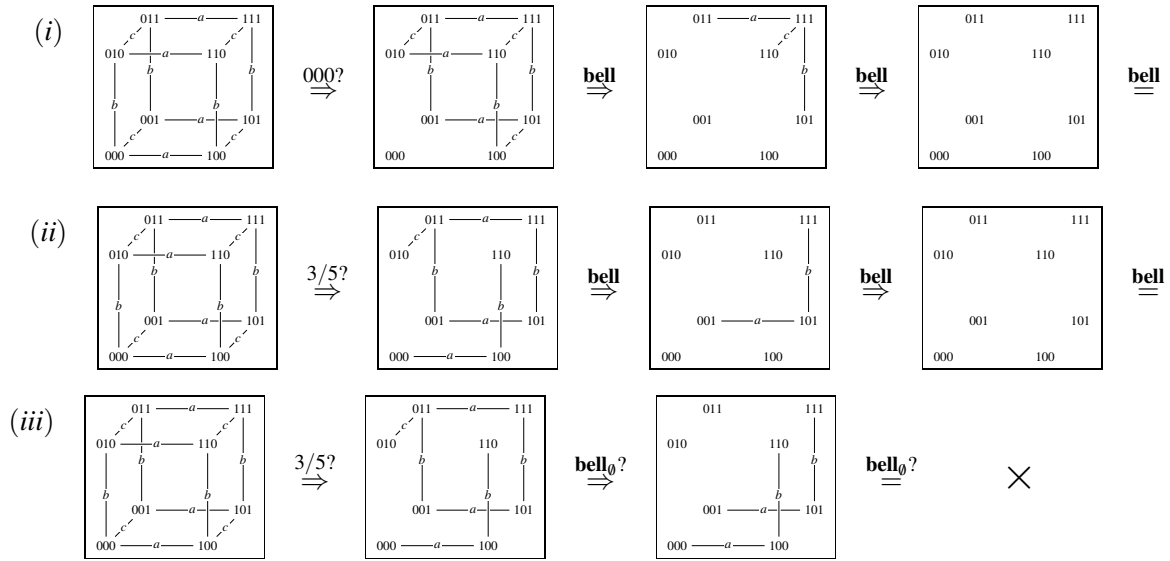

In Figure~\ref{fig.muddy}.ii the updates represent that the children who know whether they are muddy step forward. The update $\bell$ is therefore non-deterministic choice between eight different announcements of which $\bell_\emptyset$, nobody knows whether they are muddy, is only one. In Figure~\ref{fig.muddy}.iii the update $\bell_\emptyset?$ is non-deterministic choice between $\bell_\emptyset$ and its negation: somebody knows whether they are muddy. These updates are different. For example, in the three-state restriction, in $000$ only children $b$ and $c$ know whether they are muddy, in $100$ only $c$ knows that, and in $110$ only $a$ and $c$ know that. Therefore at the first $\bell$ child $a$ learns the difference between $000$ and $100$ (depending on whether $b$ steps forward) and child $b$ learns the difference between $100$ and $110$ (depending on whether child $a$ steps forward). The links are cut. Similarly for the five-state restriction of the model at the second $\bell$. Either way, eventually all children get to know whether they are muddy, and the problem is solved. But not with $\bell_\emptyset$ updates. Some links are never cut.

Now do the same in Figure~\ref{fig.muddy}.i. There, it does not matter if the updates are with $\bell_\emptyset?$ or with $\bell$.

\weg{
Something similar is the case in the two-state and six-state restrictions of (iv), although it then gets stuck, and $a$ never learns whether it is muddy.
}

\weg{
The proposition $\bell$ will be formalized as non-deterministic choice between zero / one / two / three children knowing whether they are muddy. Where again in the final stage they (should always) all know whether they are muddy. These are eight different mutually exclusive announcements of which one will always succeed. 
}

After all these years teaching logic puzzles to audiences of school children and students, it came somewhat as a surprise to us that there are really different formalizations of muddy children, and that one is more faithful to the actions of individual children stepping forward than the other one.

\paragraph{Towards M\"utzen} 
Let us now consider three children and four `colours' $0,1,2,3$, so that even if all three children have a unique colour, a child is uncertain about its own colour. The uncertainty of the agents permits a similar cube model representation consisting of $64$ states and with subdivided edges such as $000$---$a$---$100$---$a$---$200$---$a$---$300$; we will not draw it. 
The restriction on this model due to \solvable\ consists of states \[ \{000,111,222,333\} \] 
that are all singleton equivalence classes in which all three children will step forward immediately. Which is not a fixpoint. Consider the restriction consisting of \[ \{000,111,222,333\} \union \{100,120\} \] 
that contains a connected part $120$---$b$---$100$---$a$---$000$. At the announcement of $\bell$, in $111$, $222$, $333$ all children know their colour and will step forward. In $000$, $b$ and $c$ know that they are $0$, in $100$ only $c$ knows it is $0$, and in $120$ $a$ knows that it is $1$ and $c$ knows that it is $0$. And in those cases all those would then step forward. Therefore, $a$ learns the difference between $000$ (wherein $b$ knows its value) and $100$ (wherein $b$ does not know), and will step forward at the second $\bell$. Similarly, $b$ then learns the difference between $100$ and $120$. But, of course, the above restriction to $\{000,111,222,333\} \union \{100,120\}$ 
is not invariant for permutations of colours; it is different from, for example, the restriction to \[ \{000,111,222,333\} \union \{022,032\} \]
resulting from permutation $\iota(0)=2$, $\iota(1)=0$, $\iota(2)=3$, $\iota(3)=2$. 

We are not there yet, because now consider another diagonal of the 64 state cube: \[ \{033,122,211,300\} \] 
If we swap $1$ and $3$ we obtain \[ \{011,322,233,100\} \] 
which is a different restriction. Although some permutations, like swapping $1$ and $2$, leave it unchanged. Restriction $\{000,111,222,333\}$ is the only one that is invariant for \emph{any} permutation of colours, maximal, and such that all children always get to know their colour. In this case, in a boring way, because they immediately know. The complement of $\{000,111,222,333\}$ is also invariant for any permutation of colours and maximal. But the children now never get to know their colour (unlike in Figure~\ref{fig.muddy}.iii): any state $xyz$ is indistinguishable for $a$ from $vyz$ where $v \neq x,y,z$. Therefore, $\{000,111,222,333\}$ is the only restriction satisfying $\solvable$ and $\solvable'$. 

All this, for three children and four colours, is just as in M\"utzen, only there it takes many rounds. 

\weg{
\begin{figure}
\[\begin{array}{ll}
(iv) & \framebox{000 \qquad 111 \qquad 222 \qquad 333} \\ \ \\
(v) & \framebox{\text{$120$---$b$---$100$---$a$---$000$}  \qquad 111 \qquad 222 \qquad 333} \\ \ \\
(vi) & \framebox{000  \qquad 111 \qquad \text{$032$---$b$---$022$---$a$---$222$} \qquad 333} \\ \ \\
(vii) &  \framebox{033 \qquad 122 \qquad 211 \qquad 300}  \\ \ \\
(viii) & \framebox{011 \qquad 322 \qquad 233 \qquad 100} 
\end{array}\]
\caption{Different ways to solve three hatted children given four colours $0,1,2,3$. Initial model of $64$ states not shown. Equivalences classes such as $000$---$a$---$100$ given two colours $0,1$ now become equivalence classes $000$---$a$---$100$---$a$---$200$---$a$---$300$ (assume transitivity) given four colours $0,1,2,3$.}
\label{fig.colour}
\end{figure}
}

\section{M\"utzen in epistemic logic} \label{sec.formalization}

\subsection{Introduction}

We first formalize M\"utzen in standard public announcement logic where we have simplified the \solvable\ announcement to the $\solvable'$ announcement stating that no gnome wears a hat of a unique colour. We formalize it in an extended language including fixpoints. First, we survey the relevant literature.

\paragraph{Public announcement logic}
Public announcement logic \cite{plaza:1989} is a modal logic that is an extension of multi-agent epistemic logic (the logic of knowledge) \cite{hintikka:1962} with a modality expressing the consequences of all agents synchronously observing (or hearing, or receiving) new information that is considered reliable. It is an example of what are known as dynamic epistemic logics \cite{hvdetal.del:2007,hvdetal.handbook:2015}. It contains modalities $K_a \phi$ for `agent $a$ knows that $\phi$' and modalities $[\psi]\phi$ for `after truthful public announcement of $\psi$, $\phi$ (is true)'. Works such as \cite{hvdetal.del:2007,hvdetal.handbook:2015} and many others are suitable background reading.

\paragraph{Fixpoints} 
We can extend the language of public announcement logic with \emph{fixpoints}. Fixpoint semantics have been investigated for epistemic modal logics in \cite{HalpernM86,parikh1992,jfaketal.mu:2008,BaltagBF22,yanjunetal:2025} (and in an epistemically motivated but not strictly epistemic setting in \cite{bozzellietal.inf:2014,xing:2024}), where they are a more or less straightforward adaptation of those for the modal $\mu$-calculus \cite{Kozen82}. We add another inductive construct $\nu p. \phi$ to the language. In fixpoint modal logic all occurrences of $p$ in $\phi$ must be \emph{positive}: bound by (zero or) an even number of negations. In fixpoint public announcement logic it is not that simple, as we will see, and our concern is only with existence of fixpoints.

\paragraph{Arbitrary iteration of announcements} Arbitrary iteration of announcements is uncommon fare in dynamic epistemic logics, as the (satisfiability problem of the) resulting logic is undecidable and axiomatization are unknown \cite{milleretal:2005}. Arbitrary iteration will be added as yet another inductive construct $[\phi]^*\psi$ to the logical language. 

\paragraph{Public assignments} Assignments $\sigma$ are functions that, for a finite subset of all atoms (for technical reasons), map each atom $p$ in that set to a formula $\sigma(p)$ in the logical language. For $\sigma(p)=\phi$ we also write $p := \phi$. They come with yet another inductive construct $[\sigma]\phi$. See e.g.\ \cite{hvdetal.world:2008,hvdetal.aamas:2005,DitmarschHL12}. It has been frequently used in works bridging the gap between epistemic logic and applications in artificial intelligence, such as epistemic planning \cite{hvdetal.sitcalc:2011,BELLE2022103842}. (Like public assignments, fixpoints also change the value of variables in a given model, namely that of the fixpoint variable $p$ in $\nu p. \phi$.)

\subsection{Public announcement logic with fixpoints and assignments}

We succinctly define the language, structures, and semantics.

\paragraph{Language}
Given disjoint countable sets of agents and atoms (the set of agents is often required to be finite), the language of public announcement logic is inductively defined as that of the \emph{formulas} 

\[ \phi ::= p \mid \neg \phi \mid (\phi\et\psi) \mid K_a\phi \mid [\psi]\phi \mid \nu p.\phi \mid [\psi]^*\phi \mid [\sigma]\phi \] 
where $p$ is an atomic proposition and $a$ is an agent, where in clause $\nu p.\phi$ atom $p$ may have to satisfy a requirement (discussed later), and where $\sigma$ is an assignment. Any other propositional connnectives are defined by notational abbreviation, as well as the dual modalities $\M_a \phi := \neg K_a \neg \phi$, $\dia{\phi}\psi := \neg [\phi] \neg \psi$, $\dia{\phi}^*\psi := \neg [\phi]^*\neg \psi$, $\mu x.\phi := \neg \nu x. \neg \phi[x/\neg x]$ (where $\phi[x/\neg x]$ denotes uniform substitution of all \emph{free} ---not bound by another $\nu x$ occurring in $\phi$--- occurrences of $x$ in $\phi$ by $\neg x$), and non-deterministic choice between announcements $[\phi\union\psi]\chi := [\phi]\chi\et[\psi]\chi$ (and similarly for the diamond version, $\dia{\phi\union\psi}\chi := \dia{\phi}\chi\vel\dia{\psi}\chi$, which more intuitively reflects choice). No dual is needed for assignments as they are self-dual. We omit parentheses ( and ) unless this creates ambiguity. 
We also write $\phi?$ for $\phi\union\neg\phi$. 

\paragraph{Structures} The language is interpreted in epistemic models $(W,\sim,V)$ where $W$ is a set of \emph{worlds} (or \emph{states}), where $\sim$ consists of binary \emph{indistinguishability relations} $\sim_a$ between worlds for each agent $a$, that are equivalence relations, and where \emph{valuation} $V$ maps each atom to a subset of the set of worlds $W$ (these are the worlds where the atom is true). 

\paragraph{Semantics} We then define the semantics as follows, by way of a \emph{satisfiability relation} $\models$ between pairs $(M,w)$ and formulas $\phi$, where $M = (W,\sim,V)$ is an epistemic model, world $w \in W$, and where $\I{\phi}_M$ abbreviates $\{ w \in W \mid M,w \models \phi \}$. 
\[ \begin{array}{lll}
M,w \models p & \text{iff} & w \in V(p) \\
M,w \models \neg \phi & \text{iff} & M,w \not\models \phi \\
M,w \models \phi\et\psi & \text{iff} & M,w \models \phi \text{ and } M,w \models \psi \\
M,w \models K_a \phi & \text{iff} & M,v \models \phi \text{ for all $v$ such that } w \sim_a v \\
M,w \models [\psi] \phi & \text{iff} & M,w \models \psi \text{ implies } M|\psi,v \models \phi \\
M,w \models [\psi]^*\phi & \text{iff} & M,w \models [\psi]^n\phi \text{ for all } n \in \Naturals \\
M,w \models \nu p. \phi & \text{iff} & M^{p{\mapsto}U},w \models \phi \text{ where } U = \Union \{X \subseteq W \mid X \subseteq \I{\phi}_{M^{p{\mapsto}X}} \} \\ 
M,w \models [\sigma] \phi & \text{iff} & M^\sigma,w \models \phi \text{ where } M^\sigma = (W,\sim,V^\sigma) \text{ with } V^\sigma(p) = \I{\psi} \text{ for all } \sigma(p) = \psi
\end{array} \]
Here, $M|\psi$ is the \emph{restriction} (or \emph{update}) of the model $M$ to the worlds in $\I{\psi}_M \subseteq W$, defined as $M|\psi := \{W',\sim',V'\}$ where $W' = \I{\psi}_M$, ${\sim'_a} := {\sim_a \inter (W' \times W'}$ for each agent $a$, and $V'(p) := V(p) \inter W'$ for each atom $p$. Then, for arbitrary $X \subseteq W$, $M^{p{\mapsto}X} := (W,\sim,V')$ where $V'(p) = X$ and $V'(q)=V(q)$ for $q \neq p$. Furthermore, $[\psi]^n\phi$ represents iteration of the announcement $\psi$ ($[\psi]^0\phi := \phi$, $[\psi]^{n+1}\phi := [\psi][\psi]^n\phi$ and similarly for $\dia{\psi}^n\phi$). A formula is \emph{valid} if it is true in all worlds of all models. 

We can also directly give the semantics of least fixpoint namely as:
\[ \begin{array}{lll}
M,w \models \mu p. \phi & \text{iff} & M^{p{\mapsto}U},w \models \phi \text{ where } U = \Inter \{ X \subseteq W \mid \I{\phi}_{M^{p{\mapsto}X}} \subseteq X \} \\ 
\end{array} \]

The main thing about public announcement logic is that the meaning of $[\phi]\psi$ (if $\phi$ is true and announced then $\psi$ is true) is different from the meaning of $\phi \imp \psi$ (if $\phi$ is true then $\psi$ is true), because in the former case $\psi$ is interpreted in a different, updated, model. For example, $p \imp K_a p$ is invalid (not all propositions are known!) whereas $[p]K_a p$ is valid (if we make $p$ public, then $a$ knows it). The puzzling feature of the logic is that a true formula may be false after its announcement, as in the muddy children puzzle (I know whether I am muddy / I don't know whether I am muddy), and the most typical example of that is the so-called Moore-sentence \cite{moore:1942} $p \et \neg K_a p$ that is false after its announcement. It is easy to see that $[p \et \neg K_a p] \neg (p \et \neg K_a p)$ is even valid.

\subsection{Positive and negative results about fixpoints} \label{posneg}

For the fixpoint semantics we did not specify the requirement for them to exist, such as occurring positively. The issues combining announcements and fixpoints are well-known  \cite{jfaketal.mu:2008,BaltagBF22,yanjunetal:2025}. We review some of the issues and have some minor novel results.

\paragraph{What is positive?} 
How do we determine whether a fixpoint variable is `positive' in a public announcement formula? Consider again a formula of shape $[\phi]\psi$. When $\psi$ is an atom $p$, $[\phi]p$ is (always) equivalent to $\phi \imp p$, as atoms do not change their truth value in an updated model. Now consider $[\neg q]p$. This formula $[\neg q]p$ is equivalent to $\neg q \imp p$, that is, $\neg\neg q \vel p$ (or $q \vel p$), wherein the occurrence of $q$ is positive. We therefore should call the occurrence of $q$ in $[\neg q]p$ positive because its occurrence in the formula $\neg q$ of the public announcement is \emph{negative} (where negative is `not positive'). So that, for example, the greatest fixpoint $\nu q. [\neg q] p$ exists.

But we cannot simply define that $p$ is positive in $[\phi]\psi$ iff it is negative in $\phi$ and positive in $\psi$ (and that $p$ is positive in $\neg\phi$ iff it is negative in $\phi$), from which we obtain that $p$ is positive in $\dia{\phi}\psi$ iff is is positive in $\phi$ and in $\psi$. Consider $\dia{q} \M_a p$ and $\dia{q} K_a p$. Using the definitions by abbreviation and the reduction rules for public announcement \cite{hvdetal.del:2007}, the former can be rewritten to: $\dia{q} \M_a p$, iff $q \et \M_a \dia{q} p$, iff $q \et \M_a (q \et p)$. Both occurrences of $q$ are positive. No problem. On the other hand, we rewrite the latter to: $\dia{q} K_a p$, iff $\neg [q] \neg K_a p$, iff $\neg (q \imp \neg [q] K_a p)$, iff $\neg (q \imp \neg (q \imp K_a [q] p))$, iff $\neg (q \imp \neg (q \imp K_a (q \imp p)))$, which with a bit of propositional work is equivalent to $q \et K_a (\neg q \vel p)$. Problem. The second occurrence of $q$ is negative. This means that some smart strategy to expand the negation normal form 
$\phi ::= p \mid \neg p \mid \phi \vel \psi \mid \phi \et \psi \mid \M_a \phi \mid K_a \phi$
with clauses $\dia{\psi}\phi$ and $[\psi]\phi$ for announcement 
does not lead towards syntactic characterization of positivity of designated fixpoint variables in public announcement formulas.  Note that the problematic $\dia{q}K_ap$ above is in this fragment.

Still, the following holds (where $\phi(q)$ is a formula with possible occurrences of atom $q$).  
\begin{proposition} \label{prop.posss}
If a public announcement formula $\phi(q)$ is equivalent to a multi-agent epistemic formula wherein $q$ occurs positively, then the fixpoint $\nu q. \phi(q)$ exists.
\end{proposition}
The proof is obvious as any public announcement formula is equivalent to one without announcements \cite{plaza:1989}. Note that we do not say `\emph{any} equivalent epistemic formula'. For example, fixpoint $\nu q. [\neg q] p$ exists because given $[\neg q] p$, $q$ is positive in the equivalent $q \vel p$. But it is not positive in the also equivalent $(q \vel p) \et (q \vel \neg q)$.

\paragraph{The fixpoint positive fragment}
That $p$ is positive in a formula $[\phi]\psi$ when it is negative in $\phi$ is not an unknown phenomenon, as it can be shown that the (also known as) positive fragment $\phi ::= p \mid \neg p \mid (\phi \vel \psi) \mid (\phi\et\psi) \mid K_a\phi \mid [\neg\psi]\phi$ preserves truth after update \cite{hvdetal.synthese:2006,hvdetal.aimlbook:2005}. If we expand this fragment to the {\em fixpoint positive} fragment \[ \phi ::= q \mid p \mid \neg p \mid (\phi \vel \psi) \mid (\phi\et\psi) \mid K_a\phi \mid [\neg\psi]\phi \] where atoms $q$ are designated fixpoint variables and atoms $p$ are non-designated, then:
\begin{proposition}
Fixpoints $\nu q.\phi(q)$ exist for all formulas $\phi(q)$ in the fixpoint positive fragment. 
\end{proposition}
\begin{proof}
By induction on $\phi(q)$ in the  fixpoint positive fragment we extend the proof in \cite{hvdetal.synthese:2006} to show that there is an equivalent formula without public announcements wherein all occurrences of $q$ are bound by (zero or) an even number of negations. We then use Proposition~\ref{prop.posss}.
\end{proof}
We conjecture that fixpoints even exist for formulas in a further extended fixpoint positive fragment namely with an inductive clause $\nu q.\phi$. Unfortunately, such results do not get us anywhere when modelling knowledge puzzles such as Muddy Children and M\"utzen, as the puzzling aspect is a consequence of flipping truth values of announced formulas, that are therefore not in the fixpoint positive fragment.

\paragraph{Existence of fixpoints} We do not know of a syntactic characterization guaranteeing fixpoint existence in public announcement logic with fixpoints. When showing the existence of a fixpoint for M\"utzen we will therefore backtrack to directly showing their existence on the $\muetzen$ epistemic model, because, as we will see, not even monotonicity of our fixpoint formula is guaranteed for M\"utzen.

\paragraph{Assignments, announcements and fixpoints} 
A useful property of the fixpoint semantics is that for atoms $p$ and $q$, $\dia{p}\mu q. \phi(q) \eq p \et \mu q. \dia{p} \phi(q)$ is valid \cite[Fact 4]{jfaketal.mu:2008}. However, it does not hold for announcements of (all) arbitrary formulas $\psi$ instead of propositional variables $p$, and the interaction of assignments and announcements also cause additional complications \cite[Fact 5]{jfaketal.mu:2008}. Such complications are avoided in M\"utzen as our only assignments are permutations of colours over gnomes' hats.

\paragraph{Arbitrary iteration} 
Whether arbitrary iteration of announcements is definable as a fixpoint is  unclear, as in $[\phi]^*\psi$ the meaning of $\phi$ can be different at each subsequent iteration. Given the flipping of positive into negative it is also unclear whether for $[\phi]^*\psi$ one would rather have a least or a greatest fixpoint. We take the easy way out: to formalize M\"utzen we will not need arbitrary iteration, as there is an upper bound of the number of iterations of ignorance announcements in the inductive phase of solving the riddle. It is therefore sufficient to define $[\phi]^* \psi := \Et_{n=0}^{\max-1} [\phi]^n \psi$, denoted $[\phi]^{<\max} \psi$, where $\max$ is the number of colour distributions (the actual state will not be eliminated). Similarly, we define $\dia{\phi}^{<\max} \psi$ as $\Vel_{n=0}^{\max-1}\dia{\phi}^n$. In that case we will see that this is, after all, a fixpoint $\mu q. \psi \vel \dia{\phi} q$, where $q$ does not occur in $\phi$ or $\psi$. 

\subsection{Formalizing M\"utzen in public announcement logic}
 
We now can very efficiently formalize the M\"utzen riddle in public announcement logic for a set of agents that are $126$ gnomes and a set of atomic propositions $c_g$ for each gnome $g$ and each colour $c$. Our encoding needs formulas wherein all atoms occur, so that we need to restrict ourselves to a finite set of colours. For our purposes it is then enough to have just one more colour than there are gnomes, ensuring that any gnome remains uncertain about its own colour, even it if knew that all gnomes wear different colours. So we have $127$ colours. As worlds we take distributions $\delta: G \imp C$ of colours over gnomes. We let $\Delta := C^G$ denote the set of colour distributions. 
\begin{quote} {\em The epistemic model \muetzen\ consists of the domain $\Delta$ of all distributions of colours over gnomes, where two worlds $\delta$ and $\delta'$ are indistinguishable for gnome $g$ if it sees the same colours, so if for all $g' \neq g$, $\delta(g') = \delta'(g')$, and where $\delta \in V(c_g)$ if $\delta(g) = c$, that is, if $g$ wears a hat of colour $c$.} \hfill $(\muetzen)$ \end{quote}

The announcement $\solvable'$ that no gnome wears a unique colour is the formula \[ \solvable' \quad := \quad \Et_{g \in G} \Et_{c \in C} (c_g \imp \Vel_{h \in G}^{h\neq g} c_h) \]
and the announcement \bell\ of Santa ringing the Christmas bell is non-deterministic choice between mutually exclusive announcements 
$\bell := \Union_{L \subseteq G} \bell_L$ 
of formulas $\bell_L$ of $L$ gnomes \emph{L}eaving ({\bf or having left}) the room and $G{\setminus}L$ gnomes staying in the room. As leaving means knowing and staying means being ignorant, we can define: \[ \bell_L \quad := \quad \Et_{g \in L} \Et_{c \in C} (K_g c_g \vel K_g \neg c_g) \et \Et_{g \notin L} \neg\Et_{c \in C} (K_g c_g \vel K_g \neg c_g) \] The problem is solved when all know their colour, that is, when $\bell_G$ is true. To prevent Santa from having to ring the bell forever, it is preferable to distinguish rings where $L \neq G$ (there is remaining ignorance) from those where there is not, that we only want to feature once. We therefore define $\bell^{\notlast}$ as the non-deterministic announcement $\Union_{L \subset G} \bell_L$, and $\bell^\last$ as $\bell_G$. 
%
Given all this, the M\"utzen problem is to: 
\begin{quote} {\em Determine $n \in \Naturals^+$ such that $\muetzen,\delta \models \dia{\solvable'}\dia{\bell^\notlast}^{n-1}\bell^\last$.} \end{quote} At the $(n-1)$th ring of the bell, all who did not yet know the colour of their hat now learn this and will therefore leave the room at the $n$th ring of the bell. Determining $n$ requires the real colour distribution $\delta$ to become known to the problem solver. So, \emph{which} $n$ from $17$---$26$ is the correct answer? 

\subsection{Formalizing M\"utzen with a fixpoint}

We finally get to formalizing M\"utzen with a fixpoint, in the extended language. We keep the same set of propositional variables $c_g$ for gnome $g$ wearing a hat with colour $c$ but we also need some designated fixpoint variables for which we write $p,q,\dots$. 
First consider this (infelicitous) formalization of \solvable\ (determine $n \in \Naturals^+$ such that $\muetzen,\delta \models \dia{\solvable}\dia{\bell^\notlast}^{n-1}\bell^\last$): \vspace{-.3cm}
\[  \solvable \quad := \quad \nu p.\dia{p} \dia{\bell^\notlast}^* \bell^\last \quad \quad \hfill \LARGE{\text{$\times$}} \]
As expected, we need a greatest and not a least fixpoint, as under the scope of an announcement positive become negative. 
However, the announcement $\dia{p}$ binds a plethora of knowledge and ignorance formulas, so that, even if we take $\dia{\bell^\notlast}^*$ to be $\dia{\bell^\notlast}^{<\max}$, there will not be an equivalent epistemic formula obtained by eliminating announcements using the reduction axioms of public announcement logic wherein $p$ is also positive (see {\bf What is positive?} in the previous section). Could it be that maybe we do not need this to show existence? Indeed we don't, but let us postpone addressing this issue: this formalization of \solvable\ is infelicitous, because the restriction of $p$ to any singleton state in the model $\muetzen$, including those where gnomes wear unique colours, already results in all gnomes knowing their colour. 
We resolve this by enforcing that only colour distributions can be considered that are invariant for permutations of colours. For this, we have the assignments $\sigma$ and assignment modalities $[\sigma]$ at our disposal. Given a permutation $\iota: C\imp C$ of the set of colours, we denote by $\bm{\iota}$ the assignment such that for all colours $c$ and gnomes $g$, $\bm{\iota}(c_g) = \iota(c)_g$, and where $\bm{\iota}(p)=p$ for fixpoint variables $p$. Let $I$ be the set of all permutations of colours. We also extend the use of colour permutations $\iota$ in the obvious way from colours to colour distributions and sets of colour distributions. Given colour distributions $\delta,\delta'$, consider a relation $\delta\equiv_\iota\delta'$ if there is a permutation $\iota$ such that $\iota(\delta)=\delta'$. Relation $\equiv_\iota$ is an equivalence relation and induces a partition on the set of colour distributions $\Delta$. An $\equiv_\iota$ equivalence class is characterized by a partition of all gnomes into \emph{subsets of gnomes that have the same colour} (regardless of what that colour is). 

\paragraph{Formalization of solvable} This is our (felicitous) proposal for the formalization of \solvable. 
Here, given a colour distribution $\delta: G \imp C$ its \emph{description} is defined as $\bm{\delta} := \Et_{\delta(g)=c} c_g \et \Et_{\delta(g)\neq c} \neg c_g$. 
\[ \solvable \quad := \quad \nu p. \Et_{\delta\in\Delta} ((\bm{\delta}\imp p) \imp \Et_{\iota\in I} [\bm{\iota}] (\bm{\delta} \imp p)) \et \dia{p} \dia{\bell^\notlast}^* \bell^\last \]
Abbreviate the above formula \solvable\ as $\nu p. \phi(p)$. No matter how we wish to define positive, this formula is certainly not positive because of the part $\Et_{\delta\in\Delta} ((\bm{\delta}\imp p) \imp \Et_{\iota\in I} [\bm{\iota}] (\bm{\delta} \imp p))$, where $p$ is negative in the antecedent $\bm{\delta}\imp p$ of the implication. So we have to determine the existence of a fixpoint directly.

There is one more obstacle on that path. The function $f_{\phi(p)}$ mapping a subset $\Delta'\subseteq\Delta$ of colour distributions such that $\I{p}=\Delta'$ to a subset $\I{\phi(p)}=\Delta''\subseteq\Delta$ of colour distributions is {\bf not} monotonic. However, it is peculiar. In the following, let a permutation-invariant $\Delta'$ be a $\Delta'$ such that $\iota(\Delta')=\Delta'$; this includes $\Delta'=\emptyset$.
\begin{itemize}

\item For any permutation-invariant $\Delta'$ that does not contain colour distributions with unique colours we have that $f_{\phi(p)}(\Delta')=\Delta'$ so that $\Delta'\subseteq f_{\phi(p)}(\Delta')$.

\item For any permutation-invariant $\Delta'$ that contains colour distributions with unique colours we need to distinguish two cases. If $\Delta'$ only contains unique colour distributions we have that $f_{\phi(p)}(\Delta')=\emptyset$. This is because $\dia{\bell^\notlast}^* \bell^\last$ then always fails: if $\delta'\in\Delta'$ and $\delta'(g)=c$ then also $\delta''\in \Delta'$ that is as $\delta'$ except that $\delta''(g)=c'$ where $c'$ is another unique colour. So, gnome $g$ will never learn its colour. However, as also $f_{\phi(p)}(\emptyset)=\emptyset$ this does not violate monotonicity. The other case is when $\Delta'$ also contains colour distributions without unique colours. In that case $\dia{\bell^\notlast}^* \bell^\last$ still fails on the part of $\Delta'$ with unique colours, but it is not guaranteed to succeed on the part of $\Delta'$ without unique colours (for example, if for four gnomes it contains $1100$ as well as $2100$ then $a$ will now no longer find out whether it is muddy, but $b,c,d$ will --- whereas before, $a$ also would in case $1100$). Let $\Delta''$ be the part without unique colours. We then will have for some such $\Delta''$ that $\Delta''\subseteq\Delta'$ whereas $f_{\phi(p)}(\Delta'')\not\subseteq f_{\phi(p)}(\Delta')$. However, with respect to $\Delta'$, $\Delta''$ was then already maximal and all extensions of $\Delta''$ will also fail to satisfy the fixpoint formula, and all distributions in the part $\Delta'{\setminus}\Delta''$ will fail. See the next case for how to treat that.

\item For any $\Delta'$ that is not permutation invariant, we have that $f_{\phi(p)}(\Delta')=\emptyset$ because the left conjunct $\Et_{\delta\in\Delta} ((\bm{\delta}\imp p) \imp \Et_{\iota\in I} [\bm{\iota}] (\bm{\delta} \imp p))$ is now false. So that $\Delta'\not\subseteq f_{\phi(p)}(\Delta')$. This breaks monotonicity and seems to spell trouble if such a $\Delta'$ extends a permutation invariant $\Delta''$ for which $\Delta'' = f_{\phi(p)}(\Delta'')$, because we now have that $\Delta'' \subseteq \Delta'$ whereas $f_{\phi(p)}(\Delta'') \not\subseteq f_{\phi(p)}(\Delta')$. But such a $f_{\phi(p)}(\Delta'')$ does not contain members of $\Delta'{\setminus}\Delta''$, and in case $f_{\phi(p)}(\Delta'')$ was not maximal there must be a permutation invariant $\Delta'''$ such that $f_{\phi(p)}(\Delta'') \subseteq f_{\phi(p)}(\Delta''')$, where it may even be that $\Delta'' \subset \Delta' \subset \Delta'''$. (As in Section~\ref{solvableunsolvable} wherein $\{111\} \subset \{111,011\} \subset (\text{complement of } \{000\})$. In $111$ and in the seven-state restriction excluding $000$ all finally learn whether they are muddy. But not in $\{111,011\}$.)
\end{itemize} 
It is therefore still the case that the greatest fixpoint we want is the union of all permutation-invariant $\Delta'$ such that $\Delta'\subseteq f_{\phi(p)}(\Delta')$. It exists! Work done.

\paragraph{Arbitrary iteration again} Having come this far, why not consider another fixpoint, namely for the arbitrary iteration $\dia{\bell^\notlast}^* \bell^\last$, that so far we took to be the abbreviation of the maximum iteration $\dia{\bell^\notlast}^{<\max} \bell^\last$. This is actually straightforward. We can also define $\dia{\bell^\notlast}^*\bell^\last$ as the least fixpoint: \[ \mu q. \bell^\last \vel \dia{\bell^\notlast} q \] First, it is easy to see that the single occurrence of $q$ in $\bell^\last \vel \dia{\bell^\notlast} q$ remains positive in the multi-agent epistemic logic formula that we obtain by the rewriting axioms of public announcement logic. Second, the existence of an equivalent epistemic formula wherein $q$ is positive in a standard modal $\mu$-calculus therefore guarantees the existence of a least fixpoint by monotonicity. Third, in the $\muetzen$ model, that is a finite model, the iteration constructing the fixpoint will still be finite, which is an easy fix in any other sense of the word as well. But we can also easily constructively see the monotonicity because the fixpoint is the already computed extension of $p$ of $\nu p. \phi(p)$, and bottom-up we see that every singleton colour distribution satisfies $\bell^\last$. 

So we can even formalize $\solvable$ as follows. Two fixpoints! 
\[ \solvable \quad := \quad \nu p. \Et_{\delta\in\Delta} ((\bm{\delta}\imp p) \imp \Et_{\iota\in I} [\bm{\iota}] (\bm{\delta} \imp p)) \et \dia{p} (\mu q. \bell^\last \vel \dia{\bell^\notlast} q)\]
Again have that the solution of the Woeginger M\"utzen puzzle is formalized as:
\begin{quote}
Determine $n \in \Naturals^+$ such that $\muetzen,\delta \models \dia{\solvable}\dia{\bell^\notlast}^{n-1}\bell^\last$.
\end{quote}
So, merely asking once more: what is $n$? What a wonderful and creative hat problem M\"utzen is.



\paragraph*{Acknowledgements} I kindly acknowledge discussions with and encouragement or comments from Marta B\'ilkov\'a, Tim French, Malvin Gattinger, Barteld Kooi, Alexander Kurz, Roman Kuznets, Clara Lerouvillois, Yanjun Li, and Thomas {\AA}gotnes. In particular, an early stage of this research involved much interactions and discussions with Barteld Kooi during my visit to Groningen in April 2024. I also very much thank the TARK reviewers for their comments. One reviewer found an error in one of the variations of M\"utzen presented: you are correct, pink and yellow only leave at ring $9$, not at ring $8$. I acknowledge support from the {\em Wolfgang Pauli Institute} (WPI) in Vienna where I was a Fellow during the completion of this research.

\bibliographystyle{eptcs}
\bibliography{biblio2025}

\section*{Appendix: Gerhard Woeginger's Hat Problem \cite{woegingermuetzen:2013}}

Der Weihnachtsmann hat 126 Intelligenzwichtel zu einem gem\"utlichen Nachmittag mit Kaffee und Kuchen eingeladen. Als die Wichtel den Saal betreten, bekommt jeder von ihnen hinterr\"ucks eine neue Wichtelm\"utze auf den Kopf gesetzt. Das geht blitzschnell, sodass keiner von ihnen die Farbe der eigenen M\"utze zu sehen kriegt. Der Weihnachtsmann er\"offnet das Treffen mit einer kurzen Rede.

\begin{quote} {\em Meine lieben Intelligenzwichtel! Wir wollen diesen Nachmittag mit einem kleinen Denkspiel beginnen. Keiner von euch kennt die Farbe der eigenen M\"utze, und jeder von euch kann die M\"utzen aller anderen 125 Wichtel sehen. Ziel dieses Spieles ist es, die Farbe der eigenen M\"utze m\"oglichst schnell und durch reines Nachdenken herauszufinden. Ich werde nun im F\"unf-Minuten-Takt mit meiner grossen Weihnachtsglocke l\"auten. Wenn einer die eigene M\"utzenfarbe herausgefunden hat, so muss er beim n\"achsten L\"auten sofort den Saal verlassen. Im Nebenzimmer bekommt er dann eine Tasse Kaffee und ein grosses St\"uck Sachertorte serviert.} \end{quote}
Der Weihnachtsmann will gerade zur Glocke gehen, als dem Wichtel Atto eine wichtige Frage einf\"allt: \begin{quote} {\em Ja ist es denn wirklich f\"ur jeden von uns m\"oglich, seine M\"utzenfarbe durch logisches Denken zu bestimmen? Wenn zum Beispiel jeder von uns eine andere M\"utzenfarbe h\"atte, dann k\"onnte wohl niemand seine Farbe durch Denken herausfinden. Dann w\"are das Spiel f\"ur uns doch nicht zu gewinnen!} \end{quote} Der Weihnachtsmann entgegnet ihm ein wenig unwirsch: \begin{quote} {\em H\"atte, k\"onnte, w\"are!!! Nat\"urlich kann jeder von euch dieses Spiel gewinnen! Ich habe die M\"utzenfarben sehr sorgf\"altig ausgew\"ahlt, sodass jeder von euch tats\"achlich seine Farbe im Laufe des Spiels durch Denken herleiten kann.} \end{quote} Und dann beginnen die Wichtel zu denken. Und der Weihnachtsmann beginnt zu l\"auten.
\begin{itemize}
\item Beim ersten L\"auten verlassen Atto und neun andere Wichtel den Saal.
\item Beim zweiten L\"auten gehen alle Wichtel mit butterblumengelben, dotterblumengelben, \\ schl\"usselblumengelben und sonnenblumengelben M\"utzen aus dem Saal.
\item Beim dritten L\"auten gehen alle Wichtel mit karmesinroten M\"utzen, beim vierten L\"auten alle mit kaktusgr\"unen M\"utzen, beim f\"unften L\"auten alle mit aquamarinblauen M\"utzen, beim sechsten L\"auten alle mit goldorangen M\"utzen, beim siebten L\"auten alle mit bernsteinbraunen M\"utzen und beim achten L\"auten gehen alle mit muschelgrauen M\"utzen.
\item Beim neunten, zehnten, elften und zw\"olften L\"auten verl\"asst niemand den Saal.
\item Beim dreizehnten L\"auten gehen alle Wichtel mit bl\"utenweissen und alle Wichtel
mit \\ ebenholzschwarzen M\"utzen.
\end{itemize}
Und so geht es weiter. Beim $N$-ten L\"auten des Weihnachtsmanns verl\"asst schlie{\ss}lich die letzte Wichtelgruppe den Saal. Ganze sieben Mal hat der Weihnachtsmann zwischendurch gel\"autet, ohne dass jemand aus dem Saal gegangen w\"are (und bei diesen sieben Malen sind das neunte, zehnte, elfte und zw\"olfte L\"auten bereits mitgez\"ahlt). Unsere Frage lautet nun: Wie gro{\ss} ist $N$? Antwortm\"oglichkeiten:
\[\begin{array}{|ll|ll|ll|ll|ll}
1. & N = 17 \quad & 3. & N = 19 \quad & 5. & N = 21 \quad & 7. & N = 23 \quad & 9. & N = 25 \quad \\
2. & N = 18 & 4. & N = 20 & 6. & N = 22 & 8. & N = 24 & 10. & N = 26 
\end{array}\]

\section*{Appendix: Gerhard Woeginger's Hat Problem \cite{woegingermuetzen:2013} --- English}

Santa Claus invited 126 smartgnomes to a cozy afternoon with coffee and cake. When the gnomes enter the hall, each of them gets a new gnome hat placed on their head from behind. This happens at lightning speed, so that none of them can see the colour of their own hat. Santa opens the meeting with a short speech.
\begin{quote}
My dear smartgnomes! We want to start this afternoon with a little brain teaser. None of you knows the colour of your own hat, and each of you can see the hats of all the other 125 gnomes. The aim of this game is to find out the colour of your own hat as quickly as possible and through pure thinking. I will now ring my big Christmas bell every five minutes. Once someone has found out their own hat colour, he must immediately leave the hall at the next ring of the bell. In the next room he is then served a cup of coffee and a large piece of Sachertorte.
\end{quote}
Santa is just about to ring the bell when gnome Atto comes up with an important question:
\begin{quote} Is it really possible for each of us to determine the colour of our hat through logical thinking? For example, if each of us had a different colour of hat, then no one would be able to figure out what colour it is by mere deduction. Then we couldn't win the game!
\end{quote}
Santa answers him a little gruffly:
\begin{quote}
Would, could, were!!! Of course, each of you can win this game! I chose the hat colours very carefully so that each of you can actually determine their colour through thinking during the game.
\end{quote}
And then the gnomes start thinking. And Santa starts ringing the bell.
\begin{itemize}
\item At the first ring of the bell, Atto and nine other gnomes leave the hall.
\item At the second ring, all the gnomes walk out of the hall wearing buttercup yellow, kingcup yellow, primrose yellow and sunflower yellow hats.
\item At the third ring all the gnomes leave with crimson hats, at the fourth ring all with cactus green hats, at the fifth ring all with aquamarine blue hats, at the sixth ring all with orange hats, at the seventh ring all with amber brown hats and at the eighth ring of the bell all leave with shell-gray hats.
\item At the ninth, tenth, eleventh and twelfth rings no one leaves the hall.
\item At the thirteenth ring, all the gnomes leave with blossom-white hats and all the gnomes with ebony-black hats.
\end{itemize}
And so it goes on. At the Nth ring of Santa, the last group of gnomes finally leaves the hall. Santa rang a total of seven times without anyone leaving the hall (and of these seven times, the ninth, tenth, eleventh and twelfth rings are already counted). Our question now is: What is N? Possible answers:
\[\begin{array}{|ll|ll|ll|ll|ll}
1. & N = 17 \quad & 3. & N = 19 \quad & 5. & N = 21 \quad & 7. & N = 23 \quad & 9. & N = 25 \quad \\
2. & N = 18 & 4. & N = 20 & 6. & N = 22 & 8. & N = 24 & 10. & N = 26 \qquad \hfill \cite{woegingermuetzen:2013}
\end{array}\]

\end{document}